\documentclass[aps,pra,reprint,superscriptaddress,amsmath,amssymb,a4paper,twocolumn]{revtex4-2}

\usepackage{graphicx} 
\usepackage{uniinput}
\usepackage{mathtools}
\usepackage{amsthm}
\usepackage{enumitem}
\usepackage[html]{xcolor}
\definecolor{riverlane_green}{RGB}{0, 111, 98}
\definecolor{riverlane_light_green}{RGB}{0, 150, 143}
\definecolor{riverlane_orange}{RGB}{255, 117, 0}
\definecolor{riverlane_red}{RGB}{220, 68, 5}
\definecolor{riverlane_pink}{RGB}{207, 111, 127}
\usepackage{hyperref}
\hypersetup{
}
\usepackage[capitalize]{cleveref} 
\usepackage{tcolorbox}

\usepackage[frozencache,cachedir=.]{minted}
\usepackage{attachfile2}

\newcommand{\R}{\mathbb{R}}

\newcommand{\norm}[1]{\left\lVert#1\right\rVert}

\newtheorem{proposition}[subsection]{Proposition}
\newtheorem{theorem}{Theorem}

\newtheorem{remark}[subsection]{Remark}

\newcommand{\fillcolorline}[1]{%
  \textcolor{#1}{\leaders\hrule height 1pt\hfill}\hspace{0pt}%
}

\definecolor{31119180}{RGB}{31,119,180}
\definecolor{25512714}{RGB}{255,127,14}
\definecolor{4416044}{RGB}{44,160,44}
\definecolor{4416044}{RGB}{44,160,44}
\definecolor{2143940}{RGB}{214,39,40}

\begin{document}

\title{Matrix inversion polynomials for the quantum singular value transformation}
\author{Christoph Sünderhauf}
\email{christoph.sunderhauf@riverlane.com}
\affiliation{Riverlane, Cambridge, United Kingdom}
\author{Zalán Németh}
\affiliation{Riverlane, Cambridge, United Kingdom}
\affiliation{University of Cambridge, United Kingdom}
\author{Adnaan Walayat}
\affiliation{Riverlane, Cambridge, United Kingdom}
\author{Andrew Patterson}
\affiliation{Riverlane, Cambridge, United Kingdom}
\author{Bjorn K.~Berntson}
\email{bkberntson@gmail.com}
\affiliation{Riverlane, Cambridge, United Kingdom}
\affiliation{Riverlane Research, Cambridge, Massachusetts, United States}

\date{\today}

\begin{abstract}
Quantum matrix inversion with the quantum singular value transformation (QSVT) requires a polynomial approximation to $1/x$. Several methods from the literature construct polynomials that achieve the known degree complexity $\mathcal{O}(\kappa\log(\kappa/\varepsilon))$ with condition number $\kappa$ and uniform error $\varepsilon$. However, the \emph{optimal} polynomial with lowest degree for fixed error $\varepsilon$ can only be approximated numerically with the resource-intensive Remez method, leading to impractical preprocessing runtimes. Here, we derive an analytic shortcut to the optimal polynomial. Comparisons with other polynomials from the literature, based on Taylor expansion, Chebyshev iteration, and convex optimization, confirm that our result is optimal. Furthermore, for large $\kappa\log(\kappa/\varepsilon)$, our polynomial has the smallest maximum value on $[-1,1]$ of all approaches considered, leading to reduced circuit depth due to the normalization condition of QSVT. With the Python code provided, this paper will also be useful for practitioners in the field.
\end{abstract}

\maketitle


\tableofcontents

\section{Introduction}

Solving linear systems is a core subroutine for quantum computers. The optimal query complexity $\mathcal{O}(\kappa\log(1/\varepsilon))$ with condition number $\kappa$ and desired uniform error $\varepsilon$ is achieved by adiabatic solvers \cite{costa2021optimalscalingquantumlinear, jennings2025randomizedadiabaticquantumlinear,dalzell2024shortcutoptimalquantumlinear}. The quantum singular value transformation (QSVT) \cite{gilyen2019} remains popular; going beyond the capability of adiabatic solvers by providing access to the inverted matrix as a block encoding, rather than to a solution vector.

Like earlier linear combination of unitaries (LCU) \cite{Childs_2017} based algorithms, QSVT is based on an odd polynomial approximation of $1/x$ within $S(1/\kappa)$, where
\begin{equation}
S(a):=[-1,-a]\cup[a, 1],
\end{equation} which contains the range of singular values of the block encoded matrix (or eigenvalues in the Hermitian case). The number of queries to the block encoding corresponds to the degree $d$ of the polynomial which is applied to the block-encoded matrix. Suitable polynomials are known with complexity \cite{Childs_2017, gribling2}
\begin{equation}
\label{eq: complexity}
    d = \mathcal{O}\left(\kappa\log\frac{\kappa}{\varepsilon}\right).
\end{equation}

To lower the circuit length, it is important to construct a polynomial $p(x)$ with low degree. We consider the case where a uniform error $\varepsilon$,
\begin{equation}
\label{eq:intro error}
    \left\lVert  p(x) - 1/x\right\rVert_{\infty,S(1/\kappa)} \le \varepsilon,
\end{equation}
for the solution is desired. The notation refers to the infinity norm (maximum error) on the domain $S(1/\kappa)$.

Here, in \cref{sec: methods}, we study and compare various approaches to matrix inversion polynomials \cite{Childs_2017,dong2021,Dong_2022, gribling2}. The main result of this article is an analytic formula for the optimal polynomial w.r.t.~\cref{eq:intro error}, which a subset of the authors previously announced at the conference \cite{Berntson_2024}. The result is based on \cite{privalov2007}, which implicitly constructs the optimal polynomial but falls short of providing an explicit solution, error expression, or computation method. Here, we close this gap (\cref{app: optimal polynomial proof}). As the optimal polynomial, it requires lower degree than the other approaches. Moreover, with our new method it is extremely fast to compute and we include Python code (\cref{app: code}). In \cref{sec: maxima} we consider the maximum value of the polynomials in $[-1,1]$, which is important for the efficacy of QSVT which requires normalized polynomials. For large $\kappa\log(\kappa/\varepsilon)$ the optimal polynomial also shines, making it the polynomial of choice for quantum matrix inversion.

\section{Polynomial approximation methods for $1/x$}

\label{sec: methods}

\begin{figure}
    \includegraphics[width=\linewidth]{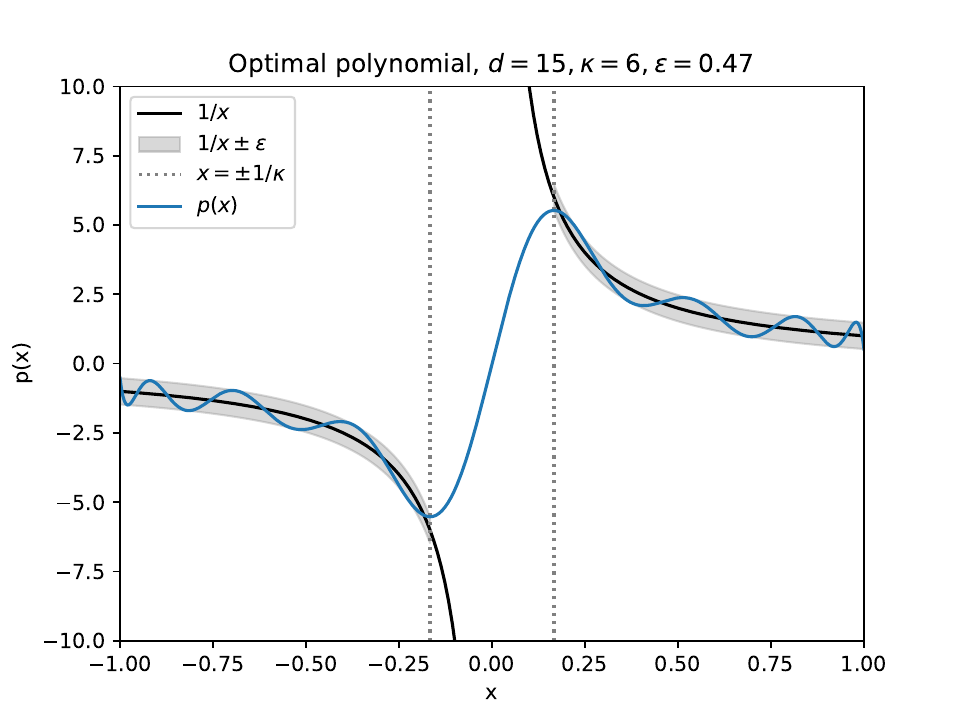}
    \caption{\label{fig:example} Example optimal polynomial \cref{eq: optimal polynomial} from \cref{sec: optimal polynomial}, approximating $1/x$ in $[-1,-1/\kappa]\cup[1/\kappa, 1]$. As the optimal polynomial, it fulfills the equioscillation theorem \cite{trefethen2019} in the shaded target region. }
\end{figure}

\subsection{Optimal polynomial}
\label{sec: optimal polynomial}
The following theorem was announced in our short conference paper \cite{Berntson_2024}, providing an analytical form for the optimal polynomial:
\begin{theorem}
\label{thm:maintext}
    Let $a = 1 / \kappa$, $a\in(0,1)$, and $d=2n-1$ odd. Then we call the polynomial $p(x)$ of degree $d$ minimizing the uniform error
    \begin{equation}
        \label{eq:optimal error def}
        \varepsilon := \left\lVert p(x)-1/x\right\rVert_{\infty,S(a)}
    \end{equation}
    in the domain $S(a) := [-1,-a]\cup[a,1]$ the ``optimal'' polynomial. It is given by
    \begin{equation}
        \label{eq: optimal polynomial}
        P_{2n-1}(x;a):=\begin{cases}\dfrac{1}{x}-\dfrac{L_n\big(\frac{2x^2-(1+a^2)}{1-a^2};a\big)}{xL_n\big(-\frac{1+a^2}{1-a^2};a\big)} & x\in \mathbb{R}\setminus\{0\} \\
        0 & x=0\end{cases}
    \end{equation}
    with
    \begin{equation}
\label{eq: optimal polynomial L}
L_n(x;a)\coloneqq \frac{1}{2^{n-1}}\bigg( T_n(x)+\frac{1-a}{1+a}T_{n-1}(x)\bigg)
\end{equation}
defined by Chebyshev polynomials. Further, the error achieved in \cref{eq:optimal error def} is
\begin{equation}
\label{eq: optimal error}
\varepsilon_{2n-1}(a)=\frac{(1-a)^n}{a(1+a)^{n-1}}.
\end{equation}
\end{theorem}
We provide a proof in \cref{app: optimal polynomial proof}, which was absent in \cite{Berntson_2024}.
See \cref{fig:example} for an example plot of the optimal polynomial.
The prefactor $1/2^{n-1}$ in \cref{eq: optimal polynomial L} is natural from the theorem stated in \cite{privalov2007} on which our proof is based. Note there was a minor misprint for $\varepsilon_{2n-1}$ in \cite{Berntson_2024}.

\begin{figure}
\includegraphics[width=\linewidth]{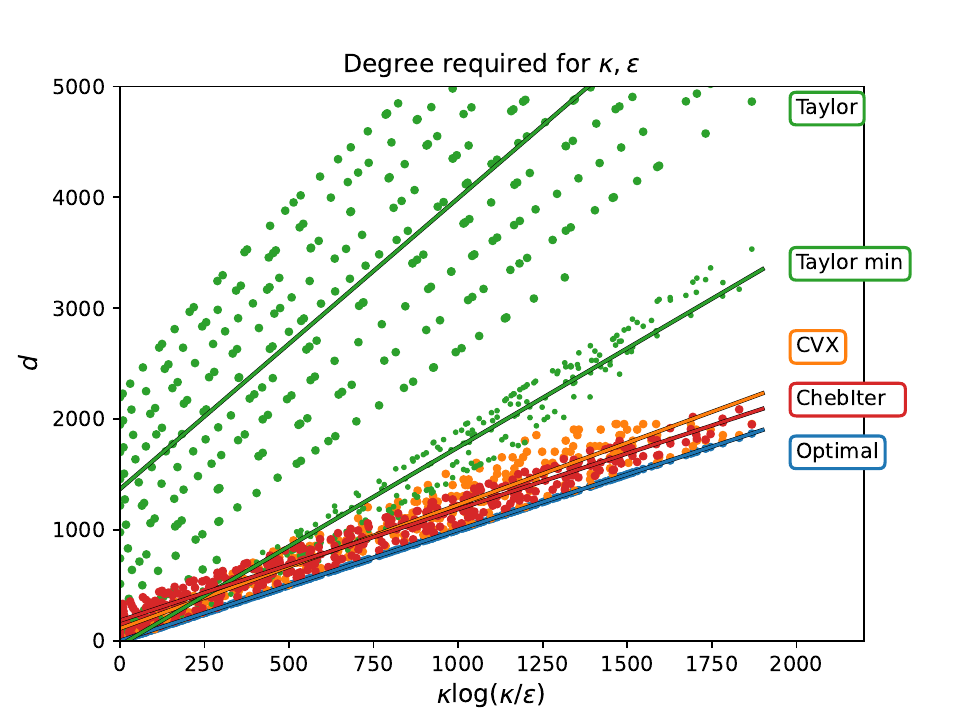}
\begin{tabular}{p{2.1cm}ll}
Polynomial & required degree (fit) & Section\\\hline
Optimal \fillcolorline{31119180} & $d\approx 1.00\,κ\log(κ/ε) +1$  & \ref{sec: optimal polynomial}\\
CVX \cite{Dong_2022}\fillcolorline{25512714} & $d\approx 1.12\,κ\log(κ/ε) +110$ & \ref{sec: convex optimisation}\\
Taylor \cite{Childs_2017} \fillcolorline{4416044} & $d\approx 2.62\,κ\log(κ/ε) +1368$  &\ref{sec: taylor}\\
Taylor min \fillcolorline{4416044} & $d\approx 1.78\,κ\log(κ/ε) -40$ & \ref{sec: taylor}\\
ChebIter  \cite{gribling2} \fillcolorline{2143940} & $d\approx 1.01\,κ\log(κ/ε) +182$ & \ref{sec:optx} \\
\end{tabular}
\caption{\label{fig:degree}Polynomial degree required for given $\kappa, \varepsilon$ for various methods. ``Taylor min'' is indicated by small dots.}
\end{figure}

The optimal polynomial described by \cref{thm:maintext} is therefore the odd polynomial minimizing the uniform error
\cref{eq:intro error}.
That is, it has minimum error $\varepsilon$ for given degree $d$, or, equivalently, minimum degree $d$ for given $\varepsilon>0$:
\Cref{eq: optimal error} can be solved for $n$ to find the degree required for given $\varepsilon, a=1/\kappa$ as
\begin{equation}
\label{eq: optimal poly degree}
n= \left\lceil\frac{ \log(1/\varepsilon) +\log(1/a)+ \log(1+a)}{ \log(1 + a) - \log(1 - a) }\right\rceil \sim \frac{1}{2}\kappa\log\frac{\kappa}{\varepsilon}.
\end{equation}
Its asymptotic behavior demonstrates that the degree $d=2n-1$ obeys the complexity \cref{eq: complexity} with constant factor 1. We show this degree required for given $κ,\varepsilon$ in \cref{fig:degree} and numerically verify the asymptotic scaling.

We found numerical evaluation of \cref{eq: optimal polynomial} to be much more stable when using a recurrence relation derived from above formulas. For this we define $\mathcal{L}_n(x; a)$ as
\begin{align}
\mathcal{L}_n(x; a) &:= \frac{L_n(x,a)}{(\alpha(a))^n},\quad
\alpha(a):=\frac{1+a}{2(1-a)},
\end{align}
which avoids the exponential growth/decay $(\alpha(a))^n$ hidden in numerator and denominator of \cref{eq: optimal polynomial}.
The optimal polynomial is
\begin{equation}
    P_{2n-1}(x;a) = \frac{1-(-1)^n\frac{(1+a)^2}{4a} \mathcal{L}_n\left(\frac{2x^2-(1+a^2)}{1-a^2};a\right)}{x},
\end{equation}
with the recurrence relation
\begin{gather}
\mathcal{L}_n(x,a) = \frac{x \mathcal{L}_{n-1}(x, a)}{\alpha(a)} - \frac{\mathcal{L}_{n-2}(x,a)}{4(\alpha(a))^2}
\end{gather}
and initial values
\begin{gather}
\mathcal{L}_1(x,a) = \frac{x + \frac{1-a}{1+a}}{\alpha(a)}, \quad 
\mathcal{L}_2(x,a) = \frac{x^2 + \frac{1-a}{2(1+a)}x - \frac{1}{2}}{(\alpha(a))^2}.
\end{gather}
It can be proven by inserting the recurrence relation of the Chebyshev polynomials into \cref{eq: optimal polynomial L} and using the identity (see \cref{prop:ascheby})
\begin{equation}\label{eq:L denominator}
L_n(-\tfrac{1+a^2}{1-a^2};a)=(-1)^n\frac{4a}{(1+a)^2} (\alpha(a))^n
\end{equation}
for the denominator of \cref{eq: optimal polynomial}. If the Chebyshev coefficients of the polynomial are desired, they can be found by evaluating at Chebyshev nodes with above recurrence and performing Chebyshev interpolation. Since it is known that $P_{2n-1}(x;a)$ is a degree $d=2n-1$ polynomial, Chebyshev interpolation is numerically exact. We provide Python code for evaluating the optimal polynomial in \cref{app: code}. \Cref{fig:runtime} shows its runtime.

\begin{figure}
\includegraphics[width=\linewidth]{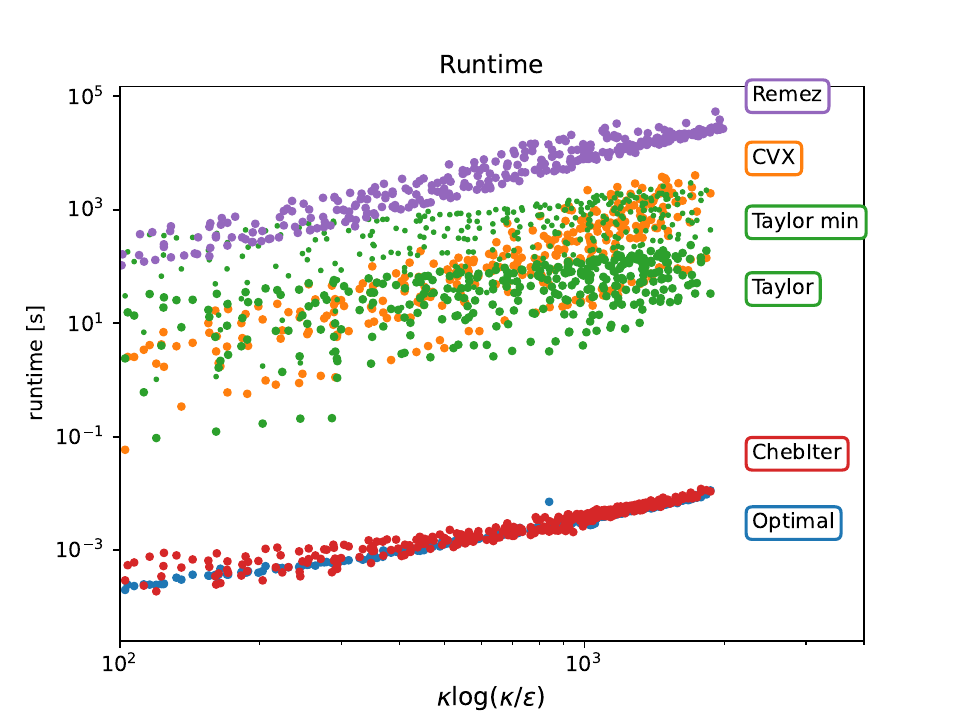}
\caption{\label{fig:runtime}Runtime of generating polynomial coefficients in the Chebyshev basis for given $\kappa,\varepsilon$ for various methods. Note that Remez and CVX actually show runtime given $\kappa,d$; a binary search required to find $d$ if $\varepsilon$ is given would increase the runtime further. ``Taylor min'' is indicated by small dots. All calculations have been performed on a Macbook M2 Pro, apart from Remez, which used an AMD EPYC 7742 processor.}
\end{figure}

\subsubsection*{Remez method}
Given a degree $d$, the optimal polynomial according to uniform error \cref{eq:optimal error def} can also be found numerically by the iterative Remez method \cite{trefethen2019}. This approach had previously been used to construct matrix inversion polynomials for QSVT, see \cite{dong2021}, whose authors had also released a software package QSPPACK \cite{qsppack_remez_notebook} including a Julia implementation of the Remez method. In practice, often the desired error $\varepsilon$ is given, and one must find the optimal polynomial with lowest degree. This required running multiple Remez methods in a binary search. The Remez method is slow and resource-hungry, not least because it requires higher arithmetic accuracy than the standard machine precision. See \cref{fig:runtime} for the runtime of a single Remez run, which can be hours and will become even longer if the target error is given, necessitating binary search. 

Our analytical shortcut in \cref{thm:maintext} provides the same optimal polynomial as the converged Remez limit, but has vastly lower runtime, by many orders of magnitude, see \cref{fig:runtime}.

\subsubsection*{Note on Chebyshev approximation}
\label{sec: cheby approximation}
A standard method of approximating functions is Chebyshev approximation \cite{trefethen2019}, which gives near-optimal uniform approximations throughout the interval $[-1,1]$. For smooth functions, the required degree is expected to be $d\sim \log(1/\varepsilon)$. In contrast, Taylor approximation is accurate near the expansion point rather than in a whole interval.
Chebyshev approximations can be found by using either continuous or discrete orthogonality relations between Chebyshev polynomials. These require integrating or sampling the desired function throughout the region $[-1,1]$. In our case, Chebyshev approximation is not directly applicable: The target function $1/x$ is only defined on $S(a)$. Moreover, for example mapping $[a,1]\mapsto[-1,1]$ would interfere with the condition of using odd Chebyshev polynomials only. 
In any case, the optimal polynomial above will have lower (or equal) error than any other polynomial approximation with the same degree.

\subsection{Convex optimization}
\label{sec: convex optimisation}

Convex optimization (CVX) as a method to approximate a function for QSVT was introduced in \cite[Section IV]{Dong_2022}, with a Matlab implementation for matrix inversion in QSPPACK \cite{qsppack-qlsp-examples}. The idea is to fix a degree $d$ and fix $N$ points $S_N \subset [-1, 1]$ (chosen as roots of Chebyshev polynomials). Then, use CVX to find polynomial coefficients minimizing $\varepsilon_{\text{cvx}} = \max_{x\in S_N\cap [a,1]} |p(x) - 1/x|$.

Our results are obtained through a Python-based implementation using the CVXPY package \cite{diamond2016cvxpy,agrawal2018rewriting}. As $\varepsilon_{\text{cvx}}$ is only a lower bound of the true error, we estimate the actual error with 
\begin{equation}
\label{eq: eps estimate}
\varepsilon = \max_{x\in\texttt{np.linspace}(a, 1, 10^5)} |p(x)-1/x|.
\end{equation}
Note that this can only underestimate the error, i.e.~overestimate the efficacy of this method, leaving our conclusion valid that the optimal polynomial (\cref{sec: optimal polynomial}) is best.
We choose $N$ by starting with $8d$ and doubling it until $\varepsilon_{\text{cvx}} \ge 0.6\varepsilon$.
In order to improve numerical stability as suggested in \cite{qsppack-qlsp-examples}, we find approximations to the scaled down function $(1/x)/(2/a)$ during CVX. Similarly as described above for the Remez method, if a target $\varepsilon$ is given, a polynomial with lowest degree $d$ can be found by binary search. However, we have found even one iteration of CVX with $d$ specified (and increasing $N$ as described above) has a long computation time, as shown in \cref{fig:runtime}.
In \cref{fig:degree} we plot the degree required for $\kappa, \varepsilon$. In fact, we have generated the data by choosing $d$ and $\kappa$, and plotting them with the resulting $\varepsilon$, as described above.

\subsection{Taylor expansion}
\label{sec: taylor}
The most common approximation for the matrix inversion polynomial in quantum algorithms literature \cite{Childs_2017,gilyen2019, martyn2021} stems from the approximating polynomial
\begin{equation}
\label{eq: lit 1/x}
    \frac{1}{x} \approx \frac{1 - (1-x^2)^b}{x}.
\end{equation}
This can be understood \cite{gribling2021} as the Taylor expansion
\begin{equation}
    \frac{1}{x'} \approx \sum_{i=0}^{b-1}(1-x')^i = \frac{1-(1-x')^b}{x'}
\end{equation}of $1/x'$ around $x'=1$ up to order $b-1\in\mathbb{N}$, suitably transformed as $x\cdot \frac{1}{x'},x'\mapsto x^2$ to give the odd polynomial \cref{eq: lit 1/x}.

As shown in \cite{Childs_2017}, its Chebyshev expansion can be truncated further to a polynomial of degree
\begin{align}
    \label{eq degree lit 1/x}
    d = 2D+1,\ D &= \left\lceil\sqrt{b\log\frac{4b}{\varepsilon/2}}\right\rceil = \mathcal{O}\left(\kappa\log\frac{\kappa}{\varepsilon}\right),\\
    b &= \left\lceil\kappa^2\log\frac{\kappa}{\varepsilon/2}\right\rceil
\end{align}
which is proven \footnote{Note that in the literature, $b$ and $D$ are described such that both approximations \cref{eq degree lit 1/x} and its truncation introduce an error $\varepsilon$, so we have halved it in \cref{eq degree lit 1/x} to keep the overall error at $\varepsilon$ in line with our convention.} to be an $\varepsilon$ approximation to $1/x$ in $[1/\kappa,1]$. The Chebyshev expansion of that degree $d$ polynomial is explicitly known:
\begin{equation}
    p(x) = 4 \sum_{j=0}^{D} (-1)^j \left[ 2^{-2b} \sum_{i=j+1}^{b} \binom{2b}{b+i} \right] T_{2j+1}(x).
\end{equation}
However, the sum in the coefficients in this formula is numerically unstable; we instead generate the polynomial directly from \cref{eq: lit 1/x} like the package pyQSP \cite{pyqsp2024}.

The runtime of generating the polynomial with the above procedure, as well as the required degree per \cref{eq degree lit 1/x} are shown in \cref{fig:runtime} and \cref{fig:degree} respectively under the label ``Taylor''. The data appears as multiple lines in the plot, one for each $\kappa$, as we are using different ranges of $\varepsilon$ for each $\kappa$ to span the same range $\kappa\log(\kappa/\varepsilon)$. 

The truncation of the polynomial according to the analytical result $D$ from \cref{eq degree lit 1/x} is too pessimistic. Further truncation for a given error $\varepsilon$ is possible. We minimize the truncation degree by binary search, estimating the error achieved by the polynomial with \cref{eq: eps estimate}. The results are indicated by ``Taylor min'' in the plots. While the binary search for best truncation does considerably lower the degree (\cref{fig:degree}), it does lead to longer runtimes (\cref{fig:runtime}).

\subsection{Chebyshev iteration}
\label{sec:optx}

In \cite{gribling2},\footnote{See arxiv:2109.04248 for an older arxiv version of \cite{gribling2}.} a different error bound than \cref{eq:intro error} was used. They considered instead the bound
\begin{equation}
    \lVert p(x)\cdot x - 1\rVert_{\infty, S(1/\kappa)} \le\varepsilon',
\end{equation}
and describe the optimal polynomial minimizing $\varepsilon'$ with respect to this bound.
It arises from Chebyshev iteration and is
\begin{equation}
\label{eq: chebyit poly}
    p(x) = \dfrac{1}{x}-\dfrac{T_n\big(\frac{2x^2-(1+a^2)}{1-a^2};a\big)}{xT_n\big(-\frac{1+a^2}{1-a^2};a\big)}.
\end{equation}
While similar to \cref{eq: optimal polynomial} it directly uses Chebyshev polynomials instead of $L_n(x)$.

The polynomial can also be used as an approximant w.r.t.~the error $\varepsilon$ definition \eqref{eq:intro error} used elsewhere, even if it loses the optimality property: As shown in \cite{gribling2021}, the polynomial achieves error $\varepsilon$ for
\begin{equation}
\label{eq: chebyit degree}
    d = 2n-1,\ n = \left\lceil \frac{1}{2}\kappa\log\frac{2\kappa}{\varepsilon}\right\rceil.
\end{equation}

We evaluate the polynomial by evaluating \cref{eq: chebyit poly} at $d+1$ points with $\texttt{scipy.special.eval\_chebyt()}$ and performing Chebyshev interpolation. The degree from \cref{eq: chebyit degree} and the runtime are shown in \cref{fig:degree} and \cref{fig:runtime} respectively labeled ``ChebIter''.

Unlike for the Taylor polynomial, it appears the analytical result \cref{eq: chebyit degree} is quite tight and cannot be improved much by minimizing $n$ for given error.

Comparing to the optimal polynomial, we find that the required degree of the ``ChebIter'' approach is higher (\cref{fig:degree}), albeit the asymptotic behavior of \cref{eq: optimal poly degree} and \cref{eq: chebyit degree} are the same. Runtime (\cref{fig:runtime}) is almost identical.

\section{Bounding polynomials}
\label{sec: maxima}

\begin{figure}
    \includegraphics[width=\linewidth]{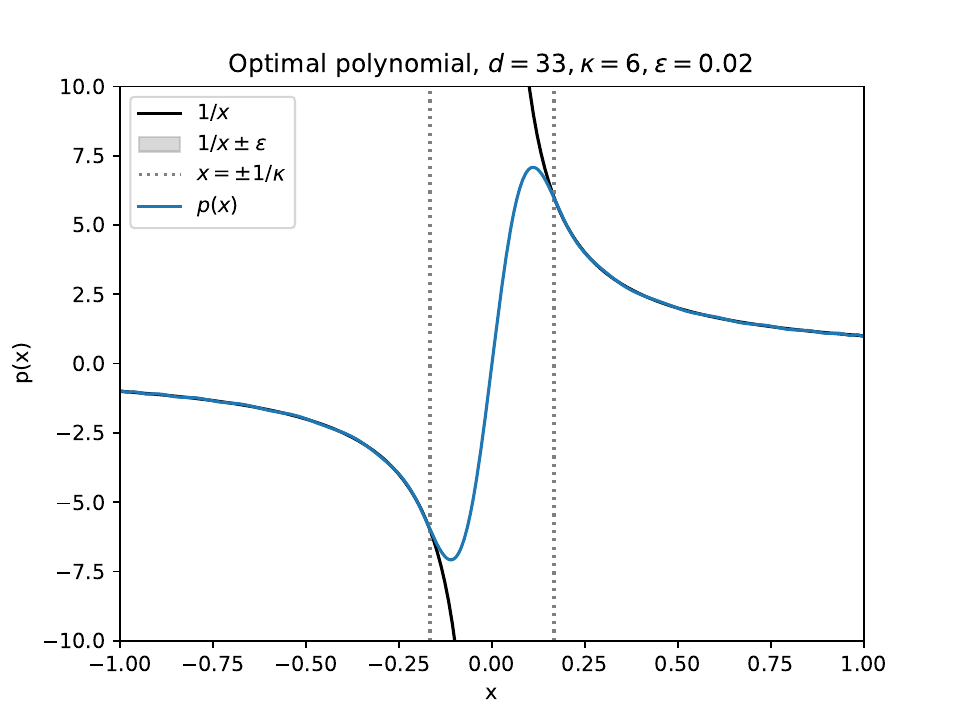}
    \caption{\label{fig:examplehigh} Example optimal polynomial \cref{eq: optimal polynomial}. Its maximum value $M>\kappa$ is attained outside the target region $[-1,-1/\kappa]\cup[1/\kappa,1]$.}
\end{figure}

The maximum
\begin{equation}
\label{eq: maximum}
M = \lVert p(x) \rVert_{\infty,[-1,1]}
\end{equation}
of the polynomials is important because QSVT requires $M\le1$.

First of all, $p(x)$ will exceed $1$ inside the region of interest $S(1/\kappa)$, with a maximal value $\lVert p(x)\rVert_{\infty,S(1/\kappa)} \le \kappa + \varepsilon$. This is fundamental and is dealt with by scaling down $p(x)\to p(x)/(\kappa+\varepsilon)$. However, in the region $[0,1/\kappa]$ the polynomial will generally attain $M>\kappa+\varepsilon$. See \cref{fig:examplehigh} for an example. Several strategies can be used to mitigate this. They result in a higher cost for larger $M$. The first strategy is to scale down the polynomial further:
\begin{equation}
    p(x) \to p(x)/M.
\end{equation}
Then, the result of QSVT can be amplified by an amplification factor $M/(\kappa+\varepsilon)$. Uniform singular value amplification \cite{gilyen2019} requires multiple queries to the result of QSVT, effectively leading to a multiplicative factor to the overall degree.
The second strategy is to multiply the polynomial with a window function that is close to 0 in $[0,1/\kappa]$, and close to 1 in $[1/\kappa, 1]$ (with a crossover region), as in \cite{Childs_2017,gilyen2019, martyn2021}. Multiplication of polynomials leads to an additive contribution to the overall degree, making this method preferable.

\begin{figure}
\includegraphics[width=\linewidth]{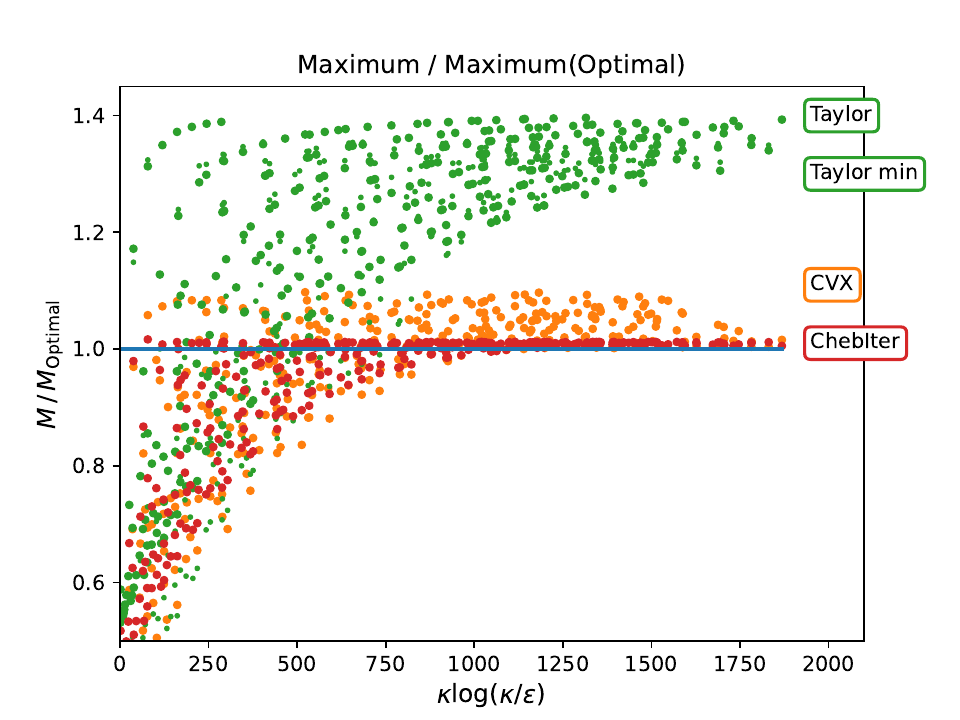}
\caption{\label{fig:maximum}Maximum value in $[-1,1]$ of generated polynomials from various methods, compared to the maximum value of the optimal polynomial. A lower maximum means a cheaper window function can be used to ensure boundedness by $1$, see \cref{sec: maxima} for details. ``Taylor min'' is indicated by small dots.}
\end{figure}

We estimate the maximum of polynomials numerically by sampling $p(x)$ at $N>d$ equidistant points $D_N=\{(2i/N) - 1\ |\  i=0,\cdots,N-1\} \subset [-1,1]$. According to \cite[Eq.~(10)]{pfister2018boundingmultivariatetrigonometricpolynomials}, the bound
\begin{equation}
|p(x)| \le \left(\cos\left(\frac{\pi}{2}\frac{d}{N}\right)\right)^{-1} \max_{x\in D_N}|p(x)|
\end{equation}
holds. Throughout we will choose $N=25d$, which ensures a tight bound 
\begin{equation}
    \max_{x\in D_N}|p(x)|\le M\le 1.002\max_{x\in D_{N=25d}}|p(x)|.
\end{equation} The advantages of this method include its feasibility (compared to analytic computation of maxima of high degree polynomials which is generally impossible), and that it cannot underestimate $M$ by an unknown and possibly large amount (compared to gradient descent which could get stuck in a local maximum). 
The numerically determined maxima are also significantly below the theoretical analytical bounds for the ``Taylor'' \cite{Childs_2017} and ``ChebIter'' \cite{gribling2021} polynomials.
\Cref{fig:maximum} shows the estimated $M\approx1.001\max_{x\in D_N}|p(x)|$, which is provably accurate up to ca.~$0.1\%$. The optimal polynomial has smallest maximum in the large $\kappa\log(\kappa/\varepsilon)$ regime. However, in the small $\kappa\log(\kappa/\varepsilon)$ regime, it could be favorable to use another polynomial depending on the degree tradeoff between $1/x$ polynomial and window function polynomial.

\section{Conclusion}
This work focused on the optimal polynomial (w.r.t.~uniform error) for matrix inversion, which reduces the degree and circuit length required in QSVT.

A comparison shows that the optimal polynomial has the lowest degree compared to other methods, commensurate with \cite{gribling2} where ``ChebIter'' was compared to the optimal polynomial for smaller parameters.
Previously this optimal polynomial could only be computed with the Remez method, leading to prohibitively large runtimes. Our exact formula and code provided are as fast as the fastest alternative approximation methods from the literature, making the optimal polynomial a viable alternative.

Furthermore, we have compared the maximum value attained by the various polynomial approximations, as it is an important factor to consider, impacting QPU runtime. We find that, for large $\kappa$ and/or small $\varepsilon$, the optimal polynomial also has the smallest maximum and is thus the polynomial of choice of the methods considered here, even when taking into account multiplication with a window  function.

Polynomials with higher degree than the optimal polynomial, but smaller maximum value exist. Depending on the exact cost contribution of the window function, such a polynomial may be favorable overall. We believe that an analytic treatment taking this tradeoff into account is not viable.

Beyond reducing the degree and maximum of the polynomial used, other factors affect the circuit length of QSVT. Notably, this includes data loading/readout \cite{Aaronson:2015scy}, like initial state preparation \cite{Low_2024, Gui_2024,O_apos_Brien_2025}, block encoding circuits \cite{Chakraborty_Block_Encoding, S_nderhauf_2024,Nguyen_2022,lapworth2025preconditionedblockencodingsquantum}, and measurement \cite{Patterson_2025}. Further, implementation details like decomposition of rotations into a universal gate set \cite{Akahoshi_2024} and QEC are also relevant to the overall performance.

Finally we would like to remark that we envisage the present polynomial to be implemented by QSVT rather than LCU, as suggested in \cite{gribling2}. The LCU would result in a (slightly) worse subnormalization, and require more ancilla qubits. Further, the only advantage of LCU was that it does not require computation of phase factors. After a recent surge of research activity \cite{Berntson_2025, alexis2024infinitequantumsignalprocessing, ni2024fastphasefactorfinding, laneve2025generalizedquantumsignalprocessing, ni2025inversenonlinearfastfourier}, computing phase factors is no longer a bottleneck, see e.g.~the software package \texttt{nlft-qsp}~\cite{laneve2025nlftqsp}.

\begin{acknowledgments}
    We thank Earl Campbell for useful discussions. This work was partially funded by Innovate UK (grant reference 10071684).
\end{acknowledgments}

\bibliography{MIP}

\begin{thebibliography}{37}%
\makeatletter
\providecommand \@ifxundefined [1]{%
 \@ifx{#1\undefined}
}%
\providecommand \@ifnum [1]{%
 \ifnum #1\expandafter \@firstoftwo
 \else \expandafter \@secondoftwo
 \fi
}%
\providecommand \@ifx [1]{%
 \ifx #1\expandafter \@firstoftwo
 \else \expandafter \@secondoftwo
 \fi
}%
\providecommand \natexlab [1]{#1}%
\providecommand \enquote  [1]{``#1''}%
\providecommand \bibnamefont  [1]{#1}%
\providecommand \bibfnamefont [1]{#1}%
\providecommand \citenamefont [1]{#1}%
\providecommand \href@noop [0]{\@secondoftwo}%
\providecommand \href [0]{\begingroup \@sanitize@url \@href}%
\providecommand \@href[1]{\@@startlink{#1}\@@href}%
\providecommand \@@href[1]{\endgroup#1\@@endlink}%
\providecommand \@sanitize@url [0]{\catcode `\\12\catcode `\$12\catcode `\&12\catcode `\#12\catcode `\^12\catcode `\_12\catcode `\%12\relax}%
\providecommand \@@startlink[1]{}%
\providecommand \@@endlink[0]{}%
\providecommand \url  [0]{\begingroup\@sanitize@url \@url }%
\providecommand \@url [1]{\endgroup\@href {#1}{\urlprefix }}%
\providecommand \urlprefix  [0]{URL }%
\providecommand \Eprint [0]{\href }%
\providecommand \doibase [0]{https://doi.org/}%
\providecommand \selectlanguage [0]{\@gobble}%
\providecommand \bibinfo  [0]{\@secondoftwo}%
\providecommand \bibfield  [0]{\@secondoftwo}%
\providecommand \translation [1]{[#1]}%
\providecommand \BibitemOpen [0]{}%
\providecommand \bibitemStop [0]{}%
\providecommand \bibitemNoStop [0]{.\EOS\space}%
\providecommand \EOS [0]{\spacefactor3000\relax}%
\providecommand \BibitemShut  [1]{\csname bibitem#1\endcsname}%
\let\auto@bib@innerbib\@empty
\bibitem [{\citenamefont {Costa}\ \emph {et~al.}(2021)\citenamefont {Costa}, \citenamefont {An}, \citenamefont {Sanders}, \citenamefont {Su}, \citenamefont {Babbush},\ and\ \citenamefont {Berry}}]{costa2021optimalscalingquantumlinear}%
  \BibitemOpen
  \bibfield  {author} {\bibinfo {author} {\bibfnamefont {P.~C.~S.}\ \bibnamefont {Costa}}, \bibinfo {author} {\bibfnamefont {D.}~\bibnamefont {An}}, \bibinfo {author} {\bibfnamefont {Y.~R.}\ \bibnamefont {Sanders}}, \bibinfo {author} {\bibfnamefont {Y.}~\bibnamefont {Su}}, \bibinfo {author} {\bibfnamefont {R.}~\bibnamefont {Babbush}},\ and\ \bibinfo {author} {\bibfnamefont {D.~W.}\ \bibnamefont {Berry}},\ }\href {https://arxiv.org/abs/2111.08152} {\bibinfo {title} {Optimal scaling quantum linear systems solver via discrete adiabatic theorem}} (\bibinfo {year} {2021}),\ \Eprint {https://arxiv.org/abs/2111.08152} {arXiv:2111.08152 [quant-ph]} \BibitemShut {NoStop}%
\bibitem [{\citenamefont {Jennings}\ \emph {et~al.}(2025)\citenamefont {Jennings}, \citenamefont {Lostaglio}, \citenamefont {Pallister}, \citenamefont {Sornborger},\ and\ \citenamefont {Subaşı}}]{jennings2025randomizedadiabaticquantumlinear}%
  \BibitemOpen
  \bibfield  {author} {\bibinfo {author} {\bibfnamefont {D.}~\bibnamefont {Jennings}}, \bibinfo {author} {\bibfnamefont {M.}~\bibnamefont {Lostaglio}}, \bibinfo {author} {\bibfnamefont {S.}~\bibnamefont {Pallister}}, \bibinfo {author} {\bibfnamefont {A.~T.}\ \bibnamefont {Sornborger}},\ and\ \bibinfo {author} {\bibfnamefont {Y.}~\bibnamefont {Subaşı}},\ }\href {https://arxiv.org/abs/2305.11352} {\bibinfo {title} {Randomized adiabatic quantum linear solver algorithm with optimal complexity scaling and detailed running costs}} (\bibinfo {year} {2025}),\ \Eprint {https://arxiv.org/abs/2305.11352} {arXiv:2305.11352 [quant-ph]} \BibitemShut {NoStop}%
\bibitem [{\citenamefont {Dalzell}(2024)}]{dalzell2024shortcutoptimalquantumlinear}%
  \BibitemOpen
  \bibfield  {author} {\bibinfo {author} {\bibfnamefont {A.~M.}\ \bibnamefont {Dalzell}},\ }\href {https://arxiv.org/abs/2406.12086} {\bibinfo {title} {A shortcut to an optimal quantum linear system solver}} (\bibinfo {year} {2024}),\ \Eprint {https://arxiv.org/abs/2406.12086} {arXiv:2406.12086 [quant-ph]} \BibitemShut {NoStop}%
\bibitem [{\citenamefont {Gily\'{e}n}\ \emph {et~al.}(2019)\citenamefont {Gily\'{e}n}, \citenamefont {Su}, \citenamefont {Low},\ and\ \citenamefont {Wiebe}}]{gilyen2019}%
  \BibitemOpen
  \bibfield  {author} {\bibinfo {author} {\bibfnamefont {A.}~\bibnamefont {Gily\'{e}n}}, \bibinfo {author} {\bibfnamefont {Y.}~\bibnamefont {Su}}, \bibinfo {author} {\bibfnamefont {G.}~\bibnamefont {Low}},\ and\ \bibinfo {author} {\bibfnamefont {N.}~\bibnamefont {Wiebe}},\ }\bibfield  {title} {\bibinfo {title} {{Quantum Singular Value Transformation and beyond: Exponential Improvements for Quantum Matrix Arithmetics}},\ }in\ \href {https://doi.org/10.1145/3313276.3316366} {\emph {\bibinfo {booktitle} {Proc. Annu. ACM Symp. Theory Comput.}}},\ \bibinfo {series and number} {STOC 2019}\ (\bibinfo  {publisher} {Association for Computing Machinery},\ \bibinfo {address} {New York, NY, USA},\ \bibinfo {year} {2019})\ pp.\ \bibinfo {pages} {193--204}\BibitemShut {NoStop}%
\bibitem [{\citenamefont {Childs}\ \emph {et~al.}(2017)\citenamefont {Childs}, \citenamefont {Kothari},\ and\ \citenamefont {Somma}}]{Childs_2017}%
  \BibitemOpen
  \bibfield  {author} {\bibinfo {author} {\bibfnamefont {A.~M.}\ \bibnamefont {Childs}}, \bibinfo {author} {\bibfnamefont {R.}~\bibnamefont {Kothari}},\ and\ \bibinfo {author} {\bibfnamefont {R.~D.}\ \bibnamefont {Somma}},\ }\bibfield  {title} {\bibinfo {title} {Quantum algorithm for systems of linear equations with exponentially improved dependence on precision},\ }\href {https://doi.org/10.1137/16m1087072} {\bibfield  {journal} {\bibinfo  {journal} {SIAM Journal on Computing}\ }\textbf {\bibinfo {volume} {46}},\ \bibinfo {pages} {1920–1950} (\bibinfo {year} {2017})}\BibitemShut {NoStop}%
\bibitem [{\citenamefont {Gribling}\ \emph {et~al.}(2024)\citenamefont {Gribling}, \citenamefont {Kerenidis},\ and\ \citenamefont {Szil\'{a}gyi}}]{gribling2}%
  \BibitemOpen
  \bibfield  {author} {\bibinfo {author} {\bibfnamefont {S.}~\bibnamefont {Gribling}}, \bibinfo {author} {\bibfnamefont {I.}~\bibnamefont {Kerenidis}},\ and\ \bibinfo {author} {\bibfnamefont {D.}~\bibnamefont {Szil\'{a}gyi}},\ }\bibfield  {title} {\bibinfo {title} {An optimal linear-combination-of-unitaries-based quantum linear system solver},\ }\bibfield  {journal} {\bibinfo  {journal} {ACM Transactions on Quantum Computing}\ }\textbf {\bibinfo {volume} {5}},\ \href {https://doi.org/10.1145/3649320} {10.1145/3649320} (\bibinfo {year} {2024})\BibitemShut {NoStop}%
\bibitem [{\citenamefont {Dong}\ \emph {et~al.}(2021)\citenamefont {Dong}, \citenamefont {Meng}, \citenamefont {Whaley},\ and\ \citenamefont {Lin}}]{dong2021}%
  \BibitemOpen
  \bibfield  {author} {\bibinfo {author} {\bibfnamefont {Y.}~\bibnamefont {Dong}}, \bibinfo {author} {\bibfnamefont {X.}~\bibnamefont {Meng}}, \bibinfo {author} {\bibfnamefont {K.~B.}\ \bibnamefont {Whaley}},\ and\ \bibinfo {author} {\bibfnamefont {L.}~\bibnamefont {Lin}},\ }\bibfield  {title} {\bibinfo {title} {Efficient phase-factor evaluation in quantum signal processing},\ }\href {https://doi.org/10.1103/PhysRevA.103.042419} {\bibfield  {journal} {\bibinfo  {journal} {Phys. Rev. A}\ }\textbf {\bibinfo {volume} {103}},\ \bibinfo {pages} {042419} (\bibinfo {year} {2021})}\BibitemShut {NoStop}%
\bibitem [{\citenamefont {Dong}\ \emph {et~al.}(2022)\citenamefont {Dong}, \citenamefont {Lin},\ and\ \citenamefont {Tong}}]{Dong_2022}%
  \BibitemOpen
  \bibfield  {author} {\bibinfo {author} {\bibfnamefont {Y.}~\bibnamefont {Dong}}, \bibinfo {author} {\bibfnamefont {L.}~\bibnamefont {Lin}},\ and\ \bibinfo {author} {\bibfnamefont {Y.}~\bibnamefont {Tong}},\ }\bibfield  {title} {\bibinfo {title} {Ground-state preparation and energy estimation on early fault-tolerant quantum computers via quantum eigenvalue transformation of unitary matrices},\ }\bibfield  {journal} {\bibinfo  {journal} {PRX Quantum}\ }\textbf {\bibinfo {volume} {3}},\ \href {https://doi.org/10.1103/prxquantum.3.040305} {10.1103/prxquantum.3.040305} (\bibinfo {year} {2022})\BibitemShut {NoStop}%
\bibitem [{\citenamefont {Berntson}\ and\ \citenamefont {Sünderhauf}(2024)}]{Berntson_2024}%
  \BibitemOpen
  \bibfield  {author} {\bibinfo {author} {\bibfnamefont {B.~K.}\ \bibnamefont {Berntson}}\ and\ \bibinfo {author} {\bibfnamefont {C.}~\bibnamefont {Sünderhauf}},\ }\bibfield  {title} {\bibinfo {title} {Two exact quantum signal processing results},\ }in\ \href {https://doi.org/10.1109/isvlsi61997.2024.00118} {\emph {\bibinfo {booktitle} {2024 IEEE Computer Society Annual Symposium on VLSI (ISVLSI)}}}\ (\bibinfo  {publisher} {IEEE},\ \bibinfo {year} {2024})\ p.\ \bibinfo {pages} {625–626}\BibitemShut {NoStop}%
\bibitem [{\citenamefont {Privalov}(2007)}]{privalov2007}%
  \BibitemOpen
  \bibfield  {author} {\bibinfo {author} {\bibfnamefont {I.}~\bibnamefont {Privalov}},\ }\bibfield  {title} {\bibinfo {title} {Approximation of $1/x$ by polynomials on $[-1,-a]\cup [a,1]$},\ }\href {https://doi.org/10.1134/S0001434607030157} {\bibfield  {journal} {\bibinfo  {journal} {Math. Notes}\ }\textbf {\bibinfo {volume} {81}},\ \bibinfo {pages} {415} (\bibinfo {year} {2007})}\BibitemShut {NoStop}%
\bibitem [{\citenamefont {Trefethen}(2019)}]{trefethen2019}%
  \BibitemOpen
  \bibfield  {author} {\bibinfo {author} {\bibfnamefont {L.}~\bibnamefont {Trefethen}},\ }\href {https://doi.org/10.1137/1.9781611975949} {\emph {\bibinfo {title} {Approximation Theory and Approximation Practice, Extended Edition}}}\ (\bibinfo  {publisher} {Society for Industrial and Applied Mathematics},\ \bibinfo {address} {Philadelphia, PA},\ \bibinfo {year} {2019})\ \Eprint {https://arxiv.org/abs/https://epubs.siam.org/doi/pdf/10.1137/1.9781611975949} {https://epubs.siam.org/doi/pdf/10.1137/1.9781611975949} \BibitemShut {NoStop}%
\bibitem [{\citenamefont {{QSPPACK contributors}}(2023)}]{qsppack_remez_notebook}%
  \BibitemOpen
  \bibfield  {author} {\bibinfo {author} {\bibnamefont {{QSPPACK contributors}}},\ }\href@noop {} {\bibinfo {title} {{Remez.ipynb} -- juliasolver notebook}},\ \bibinfo {howpublished} {\url{https://github.com/qsppack/QSPPACK/blob/aeed1c9f0b46a8a548efd5a078899c3547ed8a28/Solvers/JuliaSolver/Remez.ipynb}} (\bibinfo {year} {2023}),\ \bibinfo {note} {accessed: 2025-06-05}\BibitemShut {NoStop}%
\bibitem [{qsp(2025)}]{qsppack-qlsp-examples}%
  \BibitemOpen
  \href {https://qsppack.gitbook.io/qsppack/examples/quantum-linear-system-problems} {\bibinfo {title} {Quantum linear system problems}},\ \bibinfo {howpublished} {QSPPACK Documentation} (\bibinfo {year} {2025}),\ \bibinfo {note} {accessed: 2025-06-11}\BibitemShut {NoStop}%
\bibitem [{\citenamefont {Diamond}\ and\ \citenamefont {Boyd}(2016)}]{diamond2016cvxpy}%
  \BibitemOpen
  \bibfield  {author} {\bibinfo {author} {\bibfnamefont {S.}~\bibnamefont {Diamond}}\ and\ \bibinfo {author} {\bibfnamefont {S.}~\bibnamefont {Boyd}},\ }\bibfield  {title} {\bibinfo {title} {{CVXPY}: {A} {P}ython-embedded modeling language for convex optimization},\ }\href@noop {} {\bibfield  {journal} {\bibinfo  {journal} {Journal of Machine Learning Research}\ }\textbf {\bibinfo {volume} {17}},\ \bibinfo {pages} {1} (\bibinfo {year} {2016})}\BibitemShut {NoStop}%
\bibitem [{\citenamefont {Agrawal}\ \emph {et~al.}(2018)\citenamefont {Agrawal}, \citenamefont {Verschueren}, \citenamefont {Diamond},\ and\ \citenamefont {Boyd}}]{agrawal2018rewriting}%
  \BibitemOpen
  \bibfield  {author} {\bibinfo {author} {\bibfnamefont {A.}~\bibnamefont {Agrawal}}, \bibinfo {author} {\bibfnamefont {R.}~\bibnamefont {Verschueren}}, \bibinfo {author} {\bibfnamefont {S.}~\bibnamefont {Diamond}},\ and\ \bibinfo {author} {\bibfnamefont {S.}~\bibnamefont {Boyd}},\ }\bibfield  {title} {\bibinfo {title} {A rewriting system for convex optimization problems},\ }\href@noop {} {\bibfield  {journal} {\bibinfo  {journal} {Journal of Control and Decision}\ }\textbf {\bibinfo {volume} {5}},\ \bibinfo {pages} {42} (\bibinfo {year} {2018})}\BibitemShut {NoStop}%
\bibitem [{\citenamefont {Martyn}\ \emph {et~al.}(2021)\citenamefont {Martyn}, \citenamefont {Rossi}, \citenamefont {Tan},\ and\ \citenamefont {Chuang}}]{martyn2021}%
  \BibitemOpen
  \bibfield  {author} {\bibinfo {author} {\bibfnamefont {J.}~\bibnamefont {Martyn}}, \bibinfo {author} {\bibfnamefont {Z.}~\bibnamefont {Rossi}}, \bibinfo {author} {\bibfnamefont {A.}~\bibnamefont {Tan}},\ and\ \bibinfo {author} {\bibfnamefont {I.}~\bibnamefont {Chuang}},\ }\bibfield  {title} {\bibinfo {title} {{Grand Unification of Quantum Algorithms}},\ }\href {https://doi.org/10.1103/PRXQuantum.2.040203} {\bibfield  {journal} {\bibinfo  {journal} {PRX Quantum}\ }\textbf {\bibinfo {volume} {2}},\ \bibinfo {pages} {040203} (\bibinfo {year} {2021})}\BibitemShut {NoStop}%
\bibitem [{\citenamefont {Gribling}\ \emph {et~al.}(2021)\citenamefont {Gribling}, \citenamefont {Kerenidis},\ and\ \citenamefont {Szil\'{a}gyi}}]{gribling2021}%
  \BibitemOpen
  \bibfield  {author} {\bibinfo {author} {\bibfnamefont {S.}~\bibnamefont {Gribling}}, \bibinfo {author} {\bibfnamefont {I.}~\bibnamefont {Kerenidis}},\ and\ \bibinfo {author} {\bibfnamefont {D.}~\bibnamefont {Szil\'{a}gyi}},\ }\href {https://arxiv.org/abs/2109.04248} {\bibinfo {title} {Improving quantum linear system solvers via a gradient descent perspective}} (\bibinfo {year} {2021}),\ \Eprint {https://arxiv.org/abs/2109.04248} {arXiv:2109.04248 [quant-ph]} \BibitemShut {NoStop}%
\bibitem [{Note1()}]{Note1}%
  \BibitemOpen
  \bibinfo {note} {Note that in the literature, $b$ and $D$ are described such that both approximations \protect \cref {eq degree lit 1/x} and its truncation introduce an error $\varepsilon $, so we have halved it in \protect \cref {eq degree lit 1/x} to keep the overall error at $\varepsilon $ in line with our convention.}\BibitemShut {Stop}%
\bibitem [{\citenamefont {Chuang}\ and\ \citenamefont {contributors}(2024)}]{pyqsp2024}%
  \BibitemOpen
  \bibfield  {author} {\bibinfo {author} {\bibfnamefont {I.}~\bibnamefont {Chuang}}\ and\ \bibinfo {author} {\bibnamefont {contributors}},\ }\href@noop {} {\bibinfo {title} {pyqsp: Python package for quantum signal processing}},\ \bibinfo {howpublished} {\url{https://github.com/ichuang/pyqsp/blob/834520bad0030d260b303fb238ba3ba4899f7300/pyqsp/poly.py\#L238}} (\bibinfo {year} {2024}),\ \bibinfo {note} {accessed: 2025-06-12}\BibitemShut {NoStop}%
\bibitem [{Note2()}]{Note2}%
  \BibitemOpen
  \bibinfo {note} {See arxiv:2109.04248 for an older arxiv version of \cite {gribling2}.}\BibitemShut {Stop}%
\bibitem [{\citenamefont {Pfister}\ and\ \citenamefont {Bresler}(2018)}]{pfister2018boundingmultivariatetrigonometricpolynomials}%
  \BibitemOpen
  \bibfield  {author} {\bibinfo {author} {\bibfnamefont {L.}~\bibnamefont {Pfister}}\ and\ \bibinfo {author} {\bibfnamefont {Y.}~\bibnamefont {Bresler}},\ }\href {https://arxiv.org/abs/1802.09588} {\bibinfo {title} {Bounding multivariate trigonometric polynomials with applications to filter bank design}} (\bibinfo {year} {2018}),\ \Eprint {https://arxiv.org/abs/1802.09588} {arXiv:1802.09588 [eess.SP]} \BibitemShut {NoStop}%
\bibitem [{\citenamefont {Aaronson}(2015)}]{Aaronson:2015scy}%
  \BibitemOpen
  \bibfield  {author} {\bibinfo {author} {\bibfnamefont {S.}~\bibnamefont {Aaronson}},\ }\bibfield  {title} {\bibinfo {title} {{Read the fine print}},\ }\href {https://doi.org/10.1038/nphys3272} {\bibfield  {journal} {\bibinfo  {journal} {Nature Phys.}\ }\textbf {\bibinfo {volume} {11}},\ \bibinfo {pages} {291} (\bibinfo {year} {2015})}\BibitemShut {NoStop}%
\bibitem [{\citenamefont {Low}\ \emph {et~al.}(2024)\citenamefont {Low}, \citenamefont {Kliuchnikov},\ and\ \citenamefont {Schaeffer}}]{Low_2024}%
  \BibitemOpen
  \bibfield  {author} {\bibinfo {author} {\bibfnamefont {G.~H.}\ \bibnamefont {Low}}, \bibinfo {author} {\bibfnamefont {V.}~\bibnamefont {Kliuchnikov}},\ and\ \bibinfo {author} {\bibfnamefont {L.}~\bibnamefont {Schaeffer}},\ }\bibfield  {title} {\bibinfo {title} {Trading t gates for dirty qubits in state preparation and unitary synthesis},\ }\href {https://doi.org/10.22331/q-2024-06-17-1375} {\bibfield  {journal} {\bibinfo  {journal} {Quantum}\ }\textbf {\bibinfo {volume} {8}},\ \bibinfo {pages} {1375} (\bibinfo {year} {2024})}\BibitemShut {NoStop}%
\bibitem [{\citenamefont {Gui}\ \emph {et~al.}(2024)\citenamefont {Gui}, \citenamefont {Dalzell}, \citenamefont {Achille}, \citenamefont {Suchara},\ and\ \citenamefont {Chong}}]{Gui_2024}%
  \BibitemOpen
  \bibfield  {author} {\bibinfo {author} {\bibfnamefont {K.}~\bibnamefont {Gui}}, \bibinfo {author} {\bibfnamefont {A.~M.}\ \bibnamefont {Dalzell}}, \bibinfo {author} {\bibfnamefont {A.}~\bibnamefont {Achille}}, \bibinfo {author} {\bibfnamefont {M.}~\bibnamefont {Suchara}},\ and\ \bibinfo {author} {\bibfnamefont {F.~T.}\ \bibnamefont {Chong}},\ }\bibfield  {title} {\bibinfo {title} {Spacetime-efficient low-depth quantum state preparation with applications},\ }\href {https://doi.org/10.22331/q-2024-02-15-1257} {\bibfield  {journal} {\bibinfo  {journal} {Quantum}\ }\textbf {\bibinfo {volume} {8}},\ \bibinfo {pages} {1257} (\bibinfo {year} {2024})}\BibitemShut {NoStop}%
\bibitem [{\citenamefont {O'Brien}\ and\ \citenamefont {Sünderhauf}(2025)}]{O_apos_Brien_2025}%
  \BibitemOpen
  \bibfield  {author} {\bibinfo {author} {\bibfnamefont {O.}~\bibnamefont {O'Brien}}\ and\ \bibinfo {author} {\bibfnamefont {C.}~\bibnamefont {Sünderhauf}},\ }\bibfield  {title} {\bibinfo {title} {Quantum state preparation via piecewise {QSVT}},\ }\href {https://doi.org/10.22331/q-2025-07-03-1786} {\bibfield  {journal} {\bibinfo  {journal} {Quantum}\ }\textbf {\bibinfo {volume} {9}},\ \bibinfo {pages} {1786} (\bibinfo {year} {2025})}\BibitemShut {NoStop}%
\bibitem [{\citenamefont {Chakraborty}\ \emph {et~al.}(2019)\citenamefont {Chakraborty}, \citenamefont {Gilyén},\ and\ \citenamefont {Jeffery}}]{Chakraborty_Block_Encoding}%
  \BibitemOpen
  \bibfield  {author} {\bibinfo {author} {\bibfnamefont {S.}~\bibnamefont {Chakraborty}}, \bibinfo {author} {\bibfnamefont {A.}~\bibnamefont {Gilyén}},\ and\ \bibinfo {author} {\bibfnamefont {S.}~\bibnamefont {Jeffery}},\ }\bibfield  {title} {\bibinfo {title} {The power of block-encoded matrix powers: Improved regression techniques via faster hamiltonian simulation}\ }(\bibinfo  {publisher} {Schloss Dagstuhl – Leibniz-Zentrum für Informatik},\ \bibinfo {year} {2019})\BibitemShut {NoStop}%
\bibitem [{\citenamefont {Sünderhauf}\ \emph {et~al.}(2024)\citenamefont {Sünderhauf}, \citenamefont {Campbell},\ and\ \citenamefont {Camps}}]{S_nderhauf_2024}%
  \BibitemOpen
  \bibfield  {author} {\bibinfo {author} {\bibfnamefont {C.}~\bibnamefont {Sünderhauf}}, \bibinfo {author} {\bibfnamefont {E.}~\bibnamefont {Campbell}},\ and\ \bibinfo {author} {\bibfnamefont {J.}~\bibnamefont {Camps}},\ }\bibfield  {title} {\bibinfo {title} {Block-encoding structured matrices for data input in quantum computing},\ }\href {https://doi.org/10.22331/q-2024-01-11-1226} {\bibfield  {journal} {\bibinfo  {journal} {Quantum}\ }\textbf {\bibinfo {volume} {8}},\ \bibinfo {pages} {1226} (\bibinfo {year} {2024})}\BibitemShut {NoStop}%
\bibitem [{\citenamefont {Nguyen}\ \emph {et~al.}(2022)\citenamefont {Nguyen}, \citenamefont {Kiani},\ and\ \citenamefont {Lloyd}}]{Nguyen_2022}%
  \BibitemOpen
  \bibfield  {author} {\bibinfo {author} {\bibfnamefont {Q.~T.}\ \bibnamefont {Nguyen}}, \bibinfo {author} {\bibfnamefont {B.~T.}\ \bibnamefont {Kiani}},\ and\ \bibinfo {author} {\bibfnamefont {S.}~\bibnamefont {Lloyd}},\ }\bibfield  {title} {\bibinfo {title} {Block-encoding dense and full-rank kernels using hierarchical matrices: applications in quantum numerical linear algebra},\ }\href {https://doi.org/10.22331/q-2022-12-13-876} {\bibfield  {journal} {\bibinfo  {journal} {Quantum}\ }\textbf {\bibinfo {volume} {6}},\ \bibinfo {pages} {876} (\bibinfo {year} {2022})}\BibitemShut {NoStop}%
\bibitem [{\citenamefont {Lapworth}\ and\ \citenamefont {Sünderhauf}(2025)}]{lapworth2025preconditionedblockencodingsquantum}%
  \BibitemOpen
  \bibfield  {author} {\bibinfo {author} {\bibfnamefont {L.}~\bibnamefont {Lapworth}}\ and\ \bibinfo {author} {\bibfnamefont {C.}~\bibnamefont {Sünderhauf}},\ }\href {https://arxiv.org/abs/2502.20908} {\bibinfo {title} {Preconditioned block encodings for quantum linear systems}} (\bibinfo {year} {2025}),\ \Eprint {https://arxiv.org/abs/2502.20908} {arXiv:2502.20908 [quant-ph]} \BibitemShut {NoStop}%
\bibitem [{\citenamefont {Patterson}\ and\ \citenamefont {Lapworth}(2025)}]{Patterson_2025}%
  \BibitemOpen
  \bibfield  {author} {\bibinfo {author} {\bibfnamefont {A.}~\bibnamefont {Patterson}}\ and\ \bibinfo {author} {\bibfnamefont {L.}~\bibnamefont {Lapworth}},\ }\bibfield  {title} {\bibinfo {title} {Measurement schemes for quantum linear equation solvers},\ }\href {https://doi.org/10.1088/2058-9565/adbcd0} {\bibfield  {journal} {\bibinfo  {journal} {Quantum Science and Technology}\ }\textbf {\bibinfo {volume} {10}},\ \bibinfo {pages} {025037} (\bibinfo {year} {2025})}\BibitemShut {NoStop}%
\bibitem [{\citenamefont {Akahoshi}\ \emph {et~al.}(2024)\citenamefont {Akahoshi}, \citenamefont {Maruyama}, \citenamefont {Oshima}, \citenamefont {Sato},\ and\ \citenamefont {Fujii}}]{Akahoshi_2024}%
  \BibitemOpen
  \bibfield  {author} {\bibinfo {author} {\bibfnamefont {Y.}~\bibnamefont {Akahoshi}}, \bibinfo {author} {\bibfnamefont {K.}~\bibnamefont {Maruyama}}, \bibinfo {author} {\bibfnamefont {H.}~\bibnamefont {Oshima}}, \bibinfo {author} {\bibfnamefont {S.}~\bibnamefont {Sato}},\ and\ \bibinfo {author} {\bibfnamefont {K.}~\bibnamefont {Fujii}},\ }\bibfield  {title} {\bibinfo {title} {Partially fault-tolerant quantum computing architecture with error-corrected clifford gates and space-time efficient analog rotations},\ }\bibfield  {journal} {\bibinfo  {journal} {PRX Quantum}\ }\textbf {\bibinfo {volume} {5}},\ \href {https://doi.org/10.1103/prxquantum.5.010337} {10.1103/prxquantum.5.010337} (\bibinfo {year} {2024})\BibitemShut {NoStop}%
\bibitem [{\citenamefont {Berntson}\ and\ \citenamefont {Sünderhauf}(2025)}]{Berntson_2025}%
  \BibitemOpen
  \bibfield  {author} {\bibinfo {author} {\bibfnamefont {B.~K.}\ \bibnamefont {Berntson}}\ and\ \bibinfo {author} {\bibfnamefont {C.}~\bibnamefont {Sünderhauf}},\ }\bibfield  {title} {\bibinfo {title} {Complementary polynomials in quantum signal processing},\ }\bibfield  {journal} {\bibinfo  {journal} {Communications in Mathematical Physics}\ }\textbf {\bibinfo {volume} {406}},\ \href {https://doi.org/10.1007/s00220-025-05302-9} {10.1007/s00220-025-05302-9} (\bibinfo {year} {2025})\BibitemShut {NoStop}%
\bibitem [{\citenamefont {Alexis}\ \emph {et~al.}(2024)\citenamefont {Alexis}, \citenamefont {Lin}, \citenamefont {Mnatsakanyan}, \citenamefont {Thiele},\ and\ \citenamefont {Wang}}]{alexis2024infinitequantumsignalprocessing}%
  \BibitemOpen
  \bibfield  {author} {\bibinfo {author} {\bibfnamefont {M.}~\bibnamefont {Alexis}}, \bibinfo {author} {\bibfnamefont {L.}~\bibnamefont {Lin}}, \bibinfo {author} {\bibfnamefont {G.}~\bibnamefont {Mnatsakanyan}}, \bibinfo {author} {\bibfnamefont {C.}~\bibnamefont {Thiele}},\ and\ \bibinfo {author} {\bibfnamefont {J.}~\bibnamefont {Wang}},\ }\href {https://arxiv.org/abs/2407.05634} {\bibinfo {title} {Infinite quantum signal processing for arbitrary szego functions}} (\bibinfo {year} {2024}),\ \Eprint {https://arxiv.org/abs/2407.05634} {arXiv:2407.05634 [quant-ph]} \BibitemShut {NoStop}%
\bibitem [{\citenamefont {Ni}\ and\ \citenamefont {Ying}(2024)}]{ni2024fastphasefactorfinding}%
  \BibitemOpen
  \bibfield  {author} {\bibinfo {author} {\bibfnamefont {H.}~\bibnamefont {Ni}}\ and\ \bibinfo {author} {\bibfnamefont {L.}~\bibnamefont {Ying}},\ }\href {https://arxiv.org/abs/2410.06409} {\bibinfo {title} {Fast phase factor finding for quantum signal processing}} (\bibinfo {year} {2024}),\ \Eprint {https://arxiv.org/abs/2410.06409} {arXiv:2410.06409 [quant-ph]} \BibitemShut {NoStop}%
\bibitem [{\citenamefont {Laneve}(2025{\natexlab{a}})}]{laneve2025generalizedquantumsignalprocessing}%
  \BibitemOpen
  \bibfield  {author} {\bibinfo {author} {\bibfnamefont {L.}~\bibnamefont {Laneve}},\ }\href {https://arxiv.org/abs/2503.03026} {\bibinfo {title} {Generalized quantum signal processing and non-linear fourier transform are equivalent}} (\bibinfo {year} {2025}{\natexlab{a}}),\ \Eprint {https://arxiv.org/abs/2503.03026} {arXiv:2503.03026 [quant-ph]} \BibitemShut {NoStop}%
\bibitem [{\citenamefont {Ni}\ \emph {et~al.}(2025)\citenamefont {Ni}, \citenamefont {Sarkar}, \citenamefont {Ying},\ and\ \citenamefont {Lin}}]{ni2025inversenonlinearfastfourier}%
  \BibitemOpen
  \bibfield  {author} {\bibinfo {author} {\bibfnamefont {H.}~\bibnamefont {Ni}}, \bibinfo {author} {\bibfnamefont {R.}~\bibnamefont {Sarkar}}, \bibinfo {author} {\bibfnamefont {L.}~\bibnamefont {Ying}},\ and\ \bibinfo {author} {\bibfnamefont {L.}~\bibnamefont {Lin}},\ }\href {https://arxiv.org/abs/2505.12615} {\bibinfo {title} {Inverse nonlinear fast {Fourier} transform on {SU(2)} with applications to quantum signal processing}} (\bibinfo {year} {2025}),\ \Eprint {https://arxiv.org/abs/2505.12615} {arXiv:2505.12615 [quant-ph]} \BibitemShut {NoStop}%
\bibitem [{\citenamefont {Laneve}(2025{\natexlab{b}})}]{laneve2025nlftqsp}%
  \BibitemOpen
  \bibfield  {author} {\bibinfo {author} {\bibfnamefont {L.}~\bibnamefont {Laneve}},\ }\href@noop {} {\bibinfo {title} {nlft-qsp}},\ \bibinfo {howpublished} {\url{https://github.com/LorenzoLaneve/nlft-qsp}} (\bibinfo {year} {2025}{\natexlab{b}}),\ \bibinfo {note} {accessed: 2025-07-14}\BibitemShut {NoStop}%
\end{thebibliography}%

\appendix

\onecolumngrid

\section{Optimal approximation of the function $1/x$}

\label{app: optimal polynomial proof}

In this appendix we prove the analytic result  \cref{thm:maintext} for the optimal polynomial in \cref{sec: optimal polynomial}, based on \cite{privalov2007}. While the reference purports to find the optimal polynomial, it falls short of providing any explicit or readily computable formula, or expression for the error. To close this gap we begin with the following theorem stated in \cite{privalov2007} and attributed to S.N.~Bernstein. The notation 
\begin{equation}
    ||f(x)||_{\infty,A} := \sup_{x\in A} |f(x)|
\end{equation}
refers to the infinity norm on the set $A\subset\mathbb{R}$.
\begin{theorem}\label{thm:bernstein}
Let $N,n\in \mathbb{N}$, $N$ even, $\alpha_j\in \R\setminus [-1,1]$ $(j=1,\ldots,N-1)$, $\alpha_N\in (\R\setminus [-1,1])\cup \{\infty\}$,
\begin{equation}\label{eq:bernstein w}
w(x)\coloneqq \prod_{j=1}^{N}	\bigg(1-\frac{x}{\alpha_j}\bigg) \quad (x\in [-1,1]), 
\end{equation}
and
\begin{equation}
v_{\pm}(x)\coloneqq x\pm \sqrt{x^2-1}	\quad (x\in [-1,1]). 
\end{equation}
Moreover, supposing that
\begin{equation}\label{eq:alphatobeta}
\alpha_j=\frac12\bigg(\beta_j+\frac{1}{\beta_j}\bigg), \quad |\beta	_j|\leq 1 \quad (j=1,\ldots,N)
\end{equation}
holds, define
\begin{equation}
\Omega_{\pm}(x)\coloneqq \prod_{j=1}^{N} \sqrt{v_{\pm}(x)-\beta_j} ,	
\end{equation}
\begin{equation}\label{eq:Lm}
R_n\coloneqq \frac{1}{2^{n-1}}\prod_{j=1}^{N} \sqrt{1+\beta_j^2},
\end{equation}
and
\begin{equation}\label{eq:bernstein L}
L_{n}(x)\coloneqq \frac{R_n}{2}\bigg( 	v_+(x)^{n-N} \frac{\Omega_+(x)}{\Omega_-(x)}+v_-(x)^{n-N} \frac{\Omega_-(x)}{\Omega_+(x)}\Bigg)\sqrt{w(x)} \quad (x\in [-1,1]).
\end{equation}

Then,
\begin{equation}
\min_{\{c_j\}_{j=0}^{n-1}\subset\R} \norm{ \frac{\sum_{j=0}^{n-1} c_j x^j+x^n}{\sqrt{w(x)}} }_{\infty,[-1,1]}=R_n 
\end{equation}
and the minimizing polynomial is $L_n(x)$. 
\end{theorem}

\begin{remark}
\Cref{eq:bernstein L} is written in \cite{privalov2007} as (after correcting misprints and adapting notation)
\begin{equation}\label{eq:Pmalt}
L_{n}(x)\coloneqq \frac{R_n}{2}\bigg( 	v_+(x)^{n-N} \frac{\Omega_+(x)}{\Omega_-(x)}+ v_+(x)^{N-n}\frac{\Omega_-(x)}{\Omega_+(x)}\Bigg)\sqrt{w(x)} \quad (x\in [-1,1]);
\end{equation}
the equivalence with \cref{eq:bernstein L} is seen using $1/v_{\pm}(x)=v_{\mp}(x)$. One disadvantage of this form is that  it is not manifestly symmetric under $\pm\to \mp$. 
\end{remark}

We are now prepared to prove \cref{thm:maintext} from \cref{sec: optimal polynomial} through a sequence of propositions.
For this we specialize to the case $N=2$ with
\begin{equation}
    \alpha_1 = \frac{a^2+1}{a^2-1},\quad\alpha_2 = \infty,
\end{equation}
resulting in
\begin{equation}
    \beta_1 = \frac{a-1}{a+1},\quad \beta_2 = 0.
\end{equation}
This allows us to reduce the optimal polynomial for $1/x$ in $S(a)$ to the situation described in \cref{thm:bernstein}:

\begin{proposition}
    \label{thm:main}
    Let $n\in\mathbb{N}$ and $a\in(0,1)$. Then
    \begin{equation}
        \label{eq: app poly result}
        P_{2n-1}(x;a) = \frac{1}{x} - \frac{L_n\left(\frac{2x^2-(1+a^2)}{1-a^2};a\right)}{xL_n\left(-\frac{1+a^2}{1-a^2};a\right)}
    \end{equation}
    is the odd degree $2n-1$ polynomial minimizing
    \begin{equation}
        \varepsilon_{2n-1}(a) = \norm{\frac{1}{x} - P_{2n-1}(x)}_{\infty,S(a)}.
    \end{equation}
    Here, $L_n(x;a)$ is the polynomial from \cref{eq:bernstein L} with parameters specialized as above, and the minimum error is
    \begin{equation}
        \varepsilon_{2n-1}(a) = \frac{(1-a)^n}{a(1+a)^{n-1}}
    \end{equation}
\end{proposition}
\begin{proof}
    We map $x\in S(a) := [-1,-a]\cup[a,1]$ to $y\in[-1,1]$ to $x\in[a,1]$ with the transformations
    \begin{equation}
        y=\frac{2x^2-(1+a^2)}{1-a^2},\quad x=\sqrt{\frac{1+a^2}{2}+\frac{1-a^2}{2}y}.
    \end{equation}
    The transformation results in
    \begin{align}
        \frac{1}{x} - P_{2n-1}(x) = \frac{1 - xP_{2n-1}(x)}{\sqrt{\frac{1+a^2}{2}+\frac{1-a^2}{2}y}} = \frac{\sqrt{2}}{\sqrt{1+a^2}}\frac{q_n(y)}{\sqrt{w(y)}},
    \end{align}
    where $w(y)=1-y/\alpha_1$ from \cref{eq:bernstein w}, and 
    \begin{equation}\label{eq: q and P}
        q_n(y) := 1-xP_{2n-1}(x)
    \end{equation} is a degree $n$ polynomial in $y$, as $xP_{2n-1}(x)$ is an even degree $2n$ polynomial in $x$. Let $\bar q$ be the leading coefficient of $q_n(y)$.
    According to \cref{thm:bernstein}, 
    \begin{equation}
        \min_{\{q_n(y)\in\mathbb{R}[y]\,|\, \text{deg}(q_n(y)) = 2n\}} \norm{\frac{q_n(y)/{\bar q} }{\sqrt{w(y)}}}_{\infty,[-1,1]} = R_n
    \end{equation}
    achieves its minimum for $q_n(y)/\bar q = L_n(y)$. The coefficient $\bar q$ can be determined as $\bar q = 1/L_n(-\tfrac{1+a^2}{1-a^2})$ because $q_n(y)=1$ from \cref{eq: q and P} for $y=-\tfrac{1+a^2}{1-a^2}, x=0$. The minimum polynomial \cref{eq: app poly result} can then be computed from \cref{eq: q and P}.

    As for the error $\varepsilon$, we can calculate
    \begin{align}
        \varepsilon = \frac{\sqrt{2}}{\sqrt{1+a^2}}R_n|g| = \frac{\sqrt{2}}{\sqrt{1+a^2}}\frac{\sqrt{1+\beta_1^2}}{2^{n-1}}\left|\frac{1}{L_n(-\tfrac{1+a^2}{1-a^2})}\right|
    \end{align}
    and use \cref{prop:ascheby}.    
\end{proof}

To complete the proof of \cref{thm:maintext}, we must characterize the polynomials $L_n(x)$ defined in \cref{eq:bernstein L} and show that they can be written as in \cref{eq: optimal polynomial L}, and the term in the denominator as in \cref{eq:L denominator}.

\begin{proposition}\label{prop:recurrence}
The polynomials $L_m(x)$ defined in \cref{eq:bernstein L} satisfy the recurrence relation
\begin{equation}
L_{n+1}(x)=x L_n(x)-\frac14 L_{n-1}(x).	
\end{equation}
\end{proposition}

\begin{proof}
From \cref{eq:bernstein L}, we compute
\begin{align}\label{eq:Pm+1}
L_{n+1}(x)= &\; \frac{R_{n+1}}{2} \bigg( v_+(x)^{n+1-N}\frac{\Omega_+(x)}{\Omega_-(x)}+v_-(x)^{n+1-N}\frac{\Omega_-(x)}{\Omega_+(x)}\Bigg)\sqrt{w(x)} \nonumber\\
=&\; \frac{R_{n+1}}{2} \Bigg( \big(x+\sqrt{x^2-1}\big)v_+(x)^{n-N}\frac{\Omega_+(x)}{\Omega_-(x)}+ \big(x-\sqrt{x^2-1}\big)v_-(x)^{n-N}\frac{\Omega_-(x)}{\Omega_+(x)}	\bigg)\sqrt{w(x)}	 \nonumber\\
=&\; \frac{R_{n+1}}{2}x\Bigg( v_+(x)^{n-N}\frac{\Omega_+(x)}{\Omega_-(x)}+ v_-(x)^{n-N}\frac{\Omega_-(x)}{\Omega_+(x)}	\bigg)\sqrt{w(x)} \nonumber\\
&\; + \frac{R_{n+1}}2 \sqrt{x^2-1}\Bigg( v_+(x)^{n-N}\frac{\Omega_+(x)}{\Omega_-(x)} - v_-(x)^{n-N}\frac{\Omega_-(x)}{\Omega_+(x)}	\bigg)\sqrt{w(x)}.
\end{align}
By using that $R_{n+1}=\frac12R_n=\frac14R_{n-1}$ as a consequence of \cref{eq:Lm} and the identity
\begin{equation}
\sqrt{x^2-1}\,v_{\pm}(x)=\pm xv_{\pm}(x)\mp 1,	
\end{equation}
in \cref{eq:Pm+1}, we obtain the result. 
\end{proof}
As a consequence of \cref{prop:recurrence}, given $L_{n_0}(x)$ for any $n_0\geq2$ and any two polynomials in the set $\{L_{n_0-2}(x),L_{n_0-1}(x),L_{n_0+1}(x),L_{n_0+2}(x)\}$ (where defined) suffices to determine $L_{n}(x)$ for all $n\geq 1$. In the next proposition, we compute $L_1(x)$ and $L_2(x)$ in our special case.
\begin{proposition}
\label{prop:L1L2}
In our special case, we have:
\begin{equation}
L_1(x)=x-\beta_1,\qquad L_2(x)=x^2 -\frac12\beta_1x -\frac12.
\end{equation}
\end{proposition}
\begin{proof}
    We begin from \cref{eq:bernstein L} and compute:
\begin{align}
L_1(x)=&\;\frac{R_1}{2}\bigg(v_+(x)	^{-1}\frac{\Omega_+(x)}{\Omega_-(x)}+v_-(x)^{-1}\frac{\Omega_-(x)}{\Omega_+(x)}\bigg)\sqrt{w(x)} \nonumber\\
=&\; \frac{R_1\sqrt{w(x)}}{2}\left(v_+(x)^{-1}\frac{{\sqrt{v_+(x)-\beta_1}\sqrt{v_+(x)}}}{\sqrt{v_-(x)-\beta_1}\sqrt{v_-(x)}}+v_-(x)^{-1}\frac{{\sqrt{v_-(x)-\beta_1}\sqrt{v_-(x)}}}{\sqrt{v_+(x)-\beta_1}\sqrt{v_+(x)}}\right)\nonumber\\
&= \frac{R_1\sqrt{w(x)}}{2}\left( \frac{\sqrt{v_+(x)-\beta_1}}{\sqrt{v_-(x)-\beta_1}} + \frac{\sqrt{v_-(x)-\beta_1}}{\sqrt{v_+(x)-\beta_1}}\right) \nonumber\\
&=\frac{\sqrt{1+\beta_1^2}\sqrt{1-x/\alpha_1}}{2} \frac{(v_+(x)-\beta_1) + (v_-(x) - \beta_1)}{\sqrt{(v_-(x)-\beta_1)(v_+(x)-\beta_1)}} \nonumber\\
&=\frac{\sqrt{1+\beta^2_1}\sqrt{1-2x\beta_1/(1+\beta_1^2)}}{2} \frac{2x -2\beta_1}{\sqrt{1-2x\beta_1 + \beta_1^2}} \nonumber\\
&= x -\beta_1.
\end{align}
Similarly, we compute:
\begin{align}
    L_2(x) &= \frac{R_2}{2}\left(\frac{\Omega_+(x)}{\Omega_-(x)} + \frac{\Omega_-(x)}{\Omega_+(x)}\right)\sqrt{w(x)} \nonumber\\
    &= \frac{\sqrt{1+\beta_1}^2\sqrt{1-2x\beta_1/(1+\beta_1^2)}}{4}\left(\frac{\sqrt{v_+(x)-\beta_1}\sqrt{v_+(x)}}{\sqrt{v_-(x)-\beta_1}\sqrt{v_-(x)}}-\frac{\sqrt{v_-(x)-\beta_1}\sqrt{v_-(x)}}{\sqrt{v_+(x)-\beta_1}\sqrt{v_+(x)}}\right)\nonumber\\
    &= \frac{\sqrt{1+\beta_1^2 - 2x\beta_1}}{4}\frac{(v_+(x) - \beta_1)v_+(x) + (v_-(x)-\beta_1)v_-(x)}{\sqrt{v_-(x)-\beta_1}\sqrt{v_+(x)-\beta_1}} \nonumber\\
    &= \frac{\sqrt{1+\beta_1^2 - 2x\beta_1}}{4} \frac{2x^2 + 2\sqrt{x^2-1}^2 - 2\beta_1x}{\sqrt{1+\beta_1^2-2x\beta_1}}\nonumber\\
    &=x^2 -\frac{1}{2}\beta_1x - \frac12
\end{align}
\end{proof}

\begin{proposition}
\label{prop:ascheby}
    The following identity holds,
    \begin{equation}
L_n(x;a)\coloneqq \frac{1}{2^{n-1}}\bigg( T_n(x)+\frac{1-a}{1+a}T_{n-1}(x)\bigg).
    \end{equation}
\end{proposition}
\begin{proof}
    The equation in the proposition follows the same recurrence relation as in \cref{prop:recurrence}, and has the same initial values as in \cref{prop:L1L2}.
\end{proof}

\begin{proposition}\label{cor:Lna}
The following identity for the term in the denominator holds,
\begin{equation}
L_n(-\tfrac{1+a^2}{1-a^2};a)=	\frac{(-1)^n}{2^{n-2}}\frac{a(1+a)^{n-2}}{(1-a)^n}.
\end{equation}
\end{proposition}
\begin{proof}
    Either by induction based on the recurrence relation \cref{prop:recurrence} with starting values \cref{prop:L1L2}, or by inserting the Chebyshev polynomial identity
    \begin{equation}
        T_n(x) = \frac{1}{2}\left(v_+(x)^n + v_-(x)^n\right)
    \end{equation}
    into \cref{prop:ascheby}. We demonstrate the latter with Mathematica.

    Mathematica input:
    \vspace{-1.5ex}
    \begin{verbatim}
x = -(1 + a^2)/(1 - a^2);
vp = x + Sqrt[x^2 - 1];
vm = x - Sqrt[x^2 - 1];
t[n_] = (vp^n + vm^n)/2;
result = 1/2^(n - 1)*(t[n] + (1 - a)/(1 + a)*t[n - 1]);
Simplify[result, Assumptions -> 0 < a < 1]
    \end{verbatim}
    \vspace{-3ex}
    
    Mathematica output:
    \vspace{-1.5ex}
    \begin{verbatim}
2^(2 - n) (1/(-1 + a))^n a (1 + a)^(-2 + n)
    \end{verbatim}
\end{proof}

\section{Python code to generate the optimal polynomial}

\label{app: code}

Python file attached to pdf: \attachfile{compute_opt.py}


\begin{minted}{python}
"""Compute optimal polynomial approximating 1/x"""

import math
import numpy as np


def helper_Lfrac(n: int, x: float, a: float) -> float:
    """Compute mathcal{L}_n(x; a)"""
    alpha = (1+a)/(2*(1-a))
    l1 = (x + (1-a)/(1+a)) / alpha
    l2 = (x**2 + (1-a)/(1+a) * x / 2 - 1/2) / alpha**2
    if n == 1:
        return l1
    for _ in range(3, n+1):
        l1, l2 = l2, x * l2 / alpha - l1 / (4 * alpha**2)
    return l2


def helper_P(x: float, n: int, a: float) -> float:
    """Compute values of the polynomial P_(2n-1)(x; a)"""
    return (1 - (-1)**n * (1+a)**2 / (4*a) * helper_Lfrac(n, (2 * x**2 - (1 + a**2))/(1-a**2), a))/x


def poly(d: int, a: float) -> np.polynomial.Chebyshev:
    """Returns Chebyshev polynomial for optimal polynomial
    Args:
    d (int): odd degree
    a (float): 1/kappa for range [a,1]"""
    if d % 2 == 0:
        raise ValueError("d must be odd")
    coef = np.polynomial.chebyshev.chebinterpolate(
        helper_P, d, args=((d+1)//2, a))
    coef[0::2] = 0  # force even coefficients exactly zero
    return np.polynomial.Chebyshev(coef)


def error_for_degree(d: int, a: float) -> float:
    """Returns the poly error for degree d and a=1/kappa"""
    if d % 2 == 0:
        raise ValueError("d must be odd")
    n = (d+1)//2
    return (1-a)**n / (a * (1+a)**(n-1))


def mindegree_for_error(epsilon: float, a: float) -> int:
    """Returns the minimum degree d for a poly with error epsilon, a=1/kappa"""
    n = math.ceil((np.log(1/epsilon) + np.log(1/a) + np.log(1+a))
                  / np.log((1+a) / (1-a)))
    return 2*n-1

if __name__ == "__main__":
    # Example usage for given target epsilon
    kappa = 4
    target_epsilon = 0.1
    
    a = 1/kappa
    d = mindegree_for_error(target_epsilon, a)
    polynomial = poly(d, a)
    # actual error could be smaller than target_epsilon as d is integer
    actual_epsilon = error_for_degree(d, a)

    print("Degree {} polynomial achieving error {} in [{},1]:".format(d, actual_epsilon, a))
    print(polynomial)
\end{minted}

\end{document}